\documentclass[letterpaper,10pt,conference]{ieeeconf}

\IEEEoverridecommandlockouts                              

\usepackage{graphicx}
\usepackage{enumerate}
\usepackage{times}
\usepackage{color}
\usepackage{subfig}
\usepackage{amsmath,amssymb}

\DeclareMathOperator{\img}{\mathrm{Im}}
\newcommand{\setU}{\mathbb{U}}
\newcommand{\diag}{\mathrm{diag}}
\newcommand{\id}{\mathrm{I}}
\newcommand{\R}{\mathbb{R}}
\newcommand{\Sb}{\mathcal{S}}
\newcommand{\spn}{\mathrm{span}}
\newcommand{\D}{\displaystyle}

\def\PWM{\mathbb{PWM}}
\def\II{\textsc{ii}}
\def\I{\textsc{i}}

\def\inw{\mathrm{in}}
\def\srad{\boldsymbol{\rho}}

\newtheorem{theorem}{Theorem}
\newtheorem{lem}{Lemma}
\newtheorem{defn}{Definition}
\newtheorem{cor}{Corollary}

\title{\LARGE \bf
Large-signal stability conditions for semi-quasi-Z-source inverters:\\
switched and averaged models}

\author{Hernan Haimovich$^1$, Richard H. Middleton$^2$, and Lisandro De Nicol\'o$^1$%
\thanks{$^1$H. Haimovich and L. De Nicol\'o are with CIFASIS-CONICET and Depto. de Control, Esc. de Ing. Electr\'onica, FCEIA, Univ. Nac. de Rosario, Argentina}%
\thanks{$^2$R.H. Middleton is with the Centre for Complex Dynamic Systems and Control,
The University of Newcastle, Callaghan, NSW 2308, Australia}}

\begin{document}

\maketitle
\thispagestyle{empty}
\pagestyle{empty}

\begin{abstract}
  The recently introduced semi-quasi-Z-source inverter can be interpreted as a DC-DC converter whose input-output voltage gain may take any value in the interval $(-\infty,1)$ depending on the applied duty cycle. In order to generate a sinusoidal voltage waveform at the output of this converter, a time-varying duty cycle needs to be applied. Application of a time-varying duty cycle that produces large-signal behavior requires careful consideration of stability issues. This paper provides stability results for both the large-signal averaged and the switched models of the semi-quasi-Z-source inverter operating in continuous conduction mode. We show that if the load is linear and purely resistive then the boundedness and ultimate boundedness of the state trajectories is guaranteed provided some reasonable operation conditions are ensured. These conditions amount to keeping the duty cycle away from the extreme values $0$ or $1$ (averaged and switched models), and limiting the maximum PWM switching period (switched model). The results obtained can be used to give theoretical justification to the inverter operation strategy recently proposed by Cao et al. in \cite{5759746}.
\end{abstract}

\section{Introduction}
\label{seccIntroduccion}

The semi-quasi-Z-source inverter introduced in \cite{5759746} is a single-phase single-stage low-cost (only two active components) transformerless inverter whose input and output terminals share the same ground. This semi-quasi-Z-source inverter is especially suited for renewable-energy distributed-generation photovoltaic applications and its name derives from the Z-source inverter \cite{peng_tia03,tanxie_tpe11} because it also contains an LC network, the distinguishing feature of the Z-source inverter. However, the shoot-through state responsible for the boost capability of the Z-source inverter is not applicable to the semi-quasi-Z-source inverter and hence the principle of operation of the latter inverter is quite different.

The semi-quasi-Z-source inverter contains two active components (such as IGBTs or MOSFETs) which conduct in a complementary manner during normal operation (continuous conduction mode, CCM). The steady-state input-output voltage gain of this inverter can theoretically be any value in the interval $(-\infty,1)$, provided the circuit is operated at a constant duty cycle with a sufficiently high switching frequency. To obtain a sinusoidal voltage waveform at the output of the inverter, \cite{5759746} asserts that if the frequency of the desired output sine wave is small enough, then the steady-state input-output gain equation would be approximately valid at every time instant (after a possible initial transient) and hence the required time-varying duty cycle can be deduced from this gain equation. This inverter operation strategy was tested on a 40W prototype connected to a linear purely resistive load. Operating the inverter in this manner necessarily produces large-signal behavior, meaning that the linearized averaged model (\cite{2573579,cukthesis,ericuk_pesc82}) is not an accurate model of the evolution of the circuit variables. Also, note that the inverter operation strategy as tested on the prototype is an open-loop strategy since no feedback from the circuit variables is involved in the computation of the applied duty cycle. 

Although the inverter operation strategy proposed in \cite{5759746} was experimentally tested and showed acceptable results, the fact that such a strategy is open-loop and time-varying raises the following issues. First, extinction of the initial transient, i.e., convergence of the state trajectory to the steady-state (time-varying) one, should be established for \emph{every possible} initial condition and \emph{every possible} value of the circuit capacitors, inductors and load. Second, operation of the converter under closed-loop control, which is a necessity in order to compensate for input voltage and load variations, should be ensured not to cause unstable behavior.


In this work, we address the aforementioned issues by 
showing that if the load is linear and purely resistive then neither the averaged nor the switched models' state trajectories can become unbounded irrespective of the (time-varying) duty cycle applied (even if it involves feedback), provided the duty cycle does not reach the extreme values 0 or 1 and, for the switched model, the maximum pulse-width modulator (PWM) switching period is suitably limited. Based on these results, we show that convergence to the periodic steady-state behavior that corresponds to a periodic duty cycle (such as that proposed in \cite{5759746}) is ensured when the load is linear and purely resistive. 

For both the averaged and the switched models, our stability results are derived by exploiting the natural energy function of the circuit, even though the quadratic term of its derivative along trajectories is only negative semidefinite. In the case of the averaged model, we show how to modify the natural energy function in order to analytically construct a suitable quadratic Lyapunov function. The stability results for the averaged model provide a starting point for more complicated stability analyses, since it is known that if the switching frequency is sufficiently high, then the averaged model trajectories will be close to those of the switched model \cite{lehbas_tpe96}. In the case of the switched model, we remark that neither of the switching modes (neither of the subsystems, employing switched systems terminology \cite{liberzon_book03}) is asymptotically stable nor detectable from the output voltage, and hence stability is dependent on the limitations imposed on switching. In addition, since neither mode is asymptotically stable, then a dwell-time condition is not sufficient in order to ensure asymptotic stability.

The remainder of the paper is organized as follows. In Section~\ref{sec:inverter}, we provide a brief description of the semi-quasi-Z-source inverter, derive the switched and the averaged models, and explain the duty cycle computation method of \cite{5759746}. In Section~\ref{sec:stab-results}, we derive the stability results for the averaged model and in Section~\ref{sec:sw-model-stab} those for the switched model. Conclusions and future research directions are outlined in Section~\ref{sec:conclusions}. Most proofs are given in the Appendix.

\textbf{Notation.} For a matrix $M$, $\srad(M)$ denotes its spectral radius and $M'$ its transpose. The $i$-th column of the identity matrix is denoted $e_i$ and $\|\cdot\|$ denotes the (induced) 2-norm. 

\section{Semi-quasi-Z-source inverter}
\label{sec:inverter}

Figure~\ref{Switching model} shows the semi-quasi-Z-source inverter connected to a resistive load. The two active components are transistors $S_1$ and $S_2$. In CCM, $S_1$ and $S_2$ operate in a complementary manner so that either $S_1$ is on and $S_2$ is off (Mode I) or $S_1$ is off and $S_2$ is on (Mode II). By discontinuous conduction mode (DCM) we refer to the case when both $S_1$ and $S_2$ are on because either the antiparallel diode of $S_2$ becomes forward biased during Mode I, or that of $S_1$ during Mode II. The duty cycle of $S_1$ is denoted $d$ and hence the duty cycle of $S_2$ is $1-d$ if operating in CCM.
\begin{figure}[htb]
\begin{center}
  \def\svgwidth{7cm}
  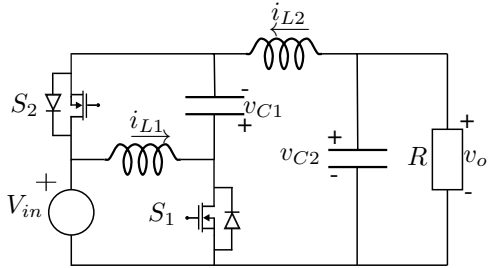
  \caption{Semi-quasi-Z-source inverter.}\label{Switching model}
\end{center}
\end{figure}

\subsection{Switched and averaged models}
\label{sec:sw-avg-models}

Defining the state vector $z$ as
\begin{equation}
  \label{eq:switchedstate}
  z := [i_{L_1} \: i_{L_2} \: v_{C_1} \: v_{C_2}]',
\end{equation}
with the positive convention for each variable as shown in Figure~\ref{Switching model}, the state equations for the switched model in CCM can readily be written as
\begin{align}
  \label{eq:1}
  \dot z(t) &= \lefteqn{[A_{\I} z(t) + b_{\I}] r(t) + [A_{\II} z(t) + b_{\II}] [1-r(t)],}\\
  \label{eq:ArArq}
  A_{i} &= P^{-1} A_{i}^q / 2, &b_{i} &= P^{-1} b_{i}^q / 2,\quad i\in\{\I,\II\},\\
  \label{eq:AIAII}
  A_{\I}^q &= \left[
  \begin{smallmatrix}
    0 & 0 & 0 & 0 \\
    0 & 0 & 1 & 1 \\
    0 & -1 & 0 & 0 \\
    0 & -1 & 0 & -1/R
  \end{smallmatrix}\right], 
& A_{\II}^q &=\left[
  \begin{smallmatrix}
    0 & 0  & -1 & 0 \\
    0 & 0  & 0  & 1 \\
    1 & 0  & 0  & 0 \\
    0 & -1 & 0  & -1/R
  \end{smallmatrix}\right],\\
  b_{\I}^q &=\left[
  \begin{smallmatrix}
    V_{\inw} & 0 & 0 & 0
  \end{smallmatrix}\right]', & b_{\II}^q &=\left[
  \begin{smallmatrix}
    0 & -V_{\inw} & 0 & 0
  \end{smallmatrix}\right]',
  \label{eq:2}\\
  \label{eq:P}
  P &:= \lefteqn{\frac{1}{2} \diag(L_1,L_2,C_1,C_2),}
\end{align}
with $r(t) = 1$ for Mode I and $r(t) = 0$ for Mode II.
These equations correspond to the case of a purely resistive load of
resistance $R$. Neither $A_{\I}= P^{-1}A_{\I}^q/2$ nor $A_{\II} =P^{-1}A_{\II}^q/2$ are
Hurwitz since they contain eigenvalues with zero real part. Using the
switched model (\ref{eq:1})--(\ref{eq:P}), the standard averaged model
is derived
\begin{equation}
  \label{eq:avgmodel}
  \dot{\langle z (t)\rangle} = [A_{\I} \langle z (t)\rangle + b_{\I}] d(t) + 
  [A_{\II} \langle z (t)\rangle + b_{\II}] [1-d(t)] 
\end{equation}
where $d(t)$ is interpreted as the instantaneous duty cycle of Mode I and is allowed to take any value between $0$ and $1$. From this averaged model, the steady-state equations for a constant duty cycle $d(t) \equiv D$ are given by
\begin{equation}
  \label{eq:avgsteady}
  0 = \left[ A_{\I} D + A_{\II} (1-D) \right] \overline{\langle z \rangle} + b_{\I} D + b_{\II} (1-D).
\end{equation}
Solving (\ref{eq:avgsteady}) yields
\begin{align}
  \overline{\langle z \rangle} &=
  \begin{bmatrix}
    \overline{\langle i_{L_1}\rangle}\qquad &
    \overline{\langle i_{L_2}\rangle}\qquad &
    \overline{\langle v_{C_1}\rangle}\qquad &
    \overline{\langle v_{C_2}\rangle} 
  \end{bmatrix}'\\
  &=
  \begin{bmatrix}
    \dfrac{\overline{\langle i_{L_2}\rangle}D}{1-D} &
    -\dfrac{\overline{\langle v_{C_2}\rangle}}{R} &
    V_{\inw}\dfrac{D}{1-D} &
    V_{\inw}\dfrac{1-2D}{1-D}
  \end{bmatrix}',\notag\\
  \label{prueba}
  \dfrac{\overline{\langle v_o \rangle}}{V_{\inw}} &= \dfrac{\overline{\langle v_{C_2} \rangle}}{V_{\inw}} = \dfrac{1-2D}{1-D}.
\end{align}
Equation (\ref{prueba}) gives the steady-state input-output voltage gain for a constant duty cycle $D$. Since the averaged models corresponding to $d(t)\equiv D=0$ or $d(t)\equiv D=1$ are not stable, then $D$ cannot reach the extreme values 0 or 1. Therefore, $D \in (0,1)$ and $\overline{\langle v_o \rangle}/V_{\inw} \in (-\infty,1)$. To obtain  $\overline{\langle v_o \rangle}\in (-V_{\inw}, V_{\inw})$ then $D\in (0,\, \frac{2}{3})$.

\subsection{Inverter operation}
\label{sec:inverter-operation}

We next briefly explain the inverter operation strategy proposed in \cite{5759746}, which is based on the steady-state gain equation (\ref{prueba}). In \cite{5759746}, the authors state that if a time-varying duty cycle is applied with a sufficiently slow time variation, then the steady-state gain equation would be approximately valid at every time instant (after the effect of the possible initial transients becomes negligible). Consequently, if the desired sinusoidal output voltage is given by 
\begin{equation}
  \label{tensionsalida}
  \langle v_{o}(t)\rangle = V_o \sin(\omega t),\\
\end{equation}
and if $\omega$ is small enough, then application of the time-varying duty cycle $d(t)$ that satisfies
\begin{equation}
  \label{eq:dutyvargain}
  \dfrac{\langle v_{o} (t)\rangle}{V_{\inw}} = \dfrac{V_o \sin(\omega t)}{V_{\inw}} = \dfrac{1-2d(t)}{1-d(t)}
\end{equation}
would approximately yield the desired sinusoidal waveform at the output. Solving (\ref{eq:dutyvargain}) for $d(t)$ yields
\begin{align}
  \label{duty}
  d(t) &=\frac{1-M\sin (\omega t)}{2-M \sin (\omega t)}\quad
  \text{where }M = \frac{V_o}{V_{\inw}}.
\end{align} 
The circuit component values should be selected so that the harmonic distortion introduced by this steady-state-equation-based operation strategy is acceptable for the required output frequency $\omega$. This operation strategy is tested on a 40W prototype \cite{5759746}, showing good results when connected to a purely resistive linear load. In the next sections, we derive results that give theoretical justification to this inverter operation strategy for all possible values of $L_1$, $L_2$, $C_1$, $C_2$, and $R$.

\section{Averaged model stability}
\label{sec:stab-results}

In this section, we will show that if the load is linear and purely resistive, then the trajectories of the averaged model will remain bounded irrespective of the duty cycle evolution, provided the duty cycle does not reach the extreme values 0 or 1. We then employ these results to show that if an open-loop duty cycle evolution, e.g. (\ref{duty}), is applied to the inverter, then the difference between any two trajectories having different initial conditions will converge to zero exponentially, and uniformly with respect to different duty cycle evolutions.


Consider the averaged model (\ref{eq:avgmodel}) and its steady-state solutions (\ref{prueba}). Let $\overline{\langle z \rangle}_{0.5}$ denote the steady-state solution corresponding to $D=0.5$:
\begin{equation}
  \label{eq:barz}
  \overline{\langle z \rangle}_{0.5} =
  \begin{bmatrix}
    0 & 0 & V_{\inw} & 0
  \end{bmatrix}',
\end{equation}
and define
\begin{align}
  \label{eq:xdef}
  x(t) &:= \langle z (t)\rangle - \overline{\langle z \rangle}_{0.5},\\
  \label{eq:mudef}
  \mu(t) &:= d(t)-0.5.
\end{align}
The averaged model in the variable $x$ can be written as
\begin{equation}
  \label{mainsystem2}
  \begin{aligned}
    \dot{x}(t) &=A(\mu(t))x(t) + B\mu(t),
  \end{aligned}
\end{equation}
where 
\begin{align}
  \label{eq:Amu}
  A&(\mu(t)) = P^{-1}\left[A_0^q + E_0^q\mu(t)\right]/2,\\
  \label{eq:B}
  B &=P^{-1}[\begin{smallmatrix}
    2V_{\inw} & 2V_{\inw} & 0 & 0
  \end{smallmatrix}]'/2,\\
  \label{eq:A0E}
  A_0^q &=\left[\begin{smallmatrix}
    0   & 0    & -0.5 & 0 \\ 
    0   & 0    & 0.5  & 1 \\ 
    0.5 & -0.5 & 0    & 0 \\ 
    0   & -1   & 0    & -1/R
  \end{smallmatrix}\right], \quad 
  E_0^q =\left[\begin{smallmatrix}
    0  & 0 & 1 & 0 \\ 
    0  & 0 & 1 & 0 \\ 
    -1 & -1& 0 & 0 \\ 
    0  & 0 & 0 & 0
  \end{smallmatrix}\right],
\end{align}
Since $d(t) \in [0,1]$, then according to (\ref{eq:mudef}) $\mu(t) \in [-0.5,0.5]$. Note that $A(-0.5)=A_{\II}$ and $A(0.5)=A_{\I}$, with $A_{\I}$ and $A_{\II}$ as in (\ref{eq:ArArq})--(\ref{eq:AIAII}). For any given $0 < \varepsilon \le 0.5$, we will consider duty cycle evolutions $d(t) \in [\varepsilon,1-\varepsilon]$, so that $\mu(t) \in \setU$ with
\begin{equation}
  \label{eq:setU}
  \setU := [-\bar\mu,\bar\mu]\quad\text{and}\quad \bar\mu := 0.5-\varepsilon. 
\end{equation}
If $V_{\inw}=0$ (passivation of the input voltage source), then the averaged model (\ref{mainsystem2}) reduces to
\begin{equation}
  \label{eq:avgVin0}
  \dot{x}(t) =A(\mu(t))x(t).
\end{equation}
We will also consider a perturbed version of system (\ref{eq:avgVin0}), given by
\begin{equation}
  \label{eq:avgpert}
  \dot{x}(t) = A(\mu(t))x(t) + Bu(t),
\end{equation}
with $B$ as in (\ref{eq:B}) and $u(t)$ an arbitrary perturbation input.

Our main result for the averaged model is the following.
\begin{theorem}
  \label{thm:boundedness}
  Let $0 < \varepsilon \le 0.5$ and consider the set $\setU$ in (\ref{eq:setU}). Let $H$ be a symmetric matrix defined as in (\ref{eq:H}) and whose nonzero entries satisfy (\ref{eq:Hentries}). 
  \begin{equation}
    \label{eq:H}
    H :=\begin{bmatrix}
      0 & 0 & h_{13} & h_{14} \\
      0 & 0 & h_{23} & h_{24} \\
      h_{13} & h_{23} & 0 & 0\\
      h_{14} & h_{24} & 0 & 0
    \end{bmatrix}
  \end{equation}
  \begin{subequations}
    \label{eq:Hentries}
    \begin{align}
      \label{eq:Hentries1}
      h_{13} &< 0,\qquad h_{14} = -\dfrac{C_{2}(h_{13}-h_{23})}{2C_{1}}\displaybreak[0]\\
      \label{eq:Hentries2}
      h_{23} &< \frac{L_{2}}{L_{1}}\frac{1-\varepsilon}{\varepsilon}h_{13} < 0\displaybreak[0]\\
      \label{eq:Hentries3}
      h_{24} &>\frac{C_{2}}{C_{1}}\frac{(h_{23}-h_{13})^2 \bar\mu^2 +
        h_{13}h_{23}}{-h_{13}4\varepsilon} > 0.
    \end{align}
  \end{subequations}
  For every $\xi > 0$, consider the quadratic function 
  \begin{equation}
    \label{eq:ptil}
    V_\xi(x) = x' \tilde{P}_\xi x,\quad\text{with}\quad \tilde{P}_\xi = P + \xi H
  \end{equation}
  and $P$ as in (\ref{eq:P}). Then,
  \begin{enumerate}[a)]
  \item for every $\xi > 0$ sufficiently small and all $\mu \in
    \setU$, we have $\tilde{P}_{\xi} > 0$ and
    \begin{equation}
      \label{eq:VdotVin0}
      A(\mu)'\tilde{P}_{\xi} + \tilde{P}_{\xi} A(\mu) \le -\alpha(\xi) \id < 0.
    \end{equation}
    Consequently, the origin of system~(\ref{eq:avgVin0}) is globally
    uniformly exponentially stable, where uniformity is with respect
    to the duty cycle $\mu$.\label{item:defneg}
  \item The trajectories of the perturbed system~(\ref{eq:avgpert}) satisfy
    \begin{equation}
      \label{eq:avgISSbnd}
      \|x(t)\| \le K \|x(0)\| e^{-\lambda t} + G \sup_{0 \le \tau \le t} \|u(\tau)\|
    \end{equation}
    for some positive constants $K$, $\lambda$ and $G$, whenever $\mu(t) \in \setU$ for all $t\ge 0$. Consequently, the system~(\ref{eq:avgpert}) is input-to-state stable (ISS) with respect to the input $u$, with exponential decay rate and linear input-to-state asymptotic gain, uniformly with respect to the duty cycle $\mu$.\label{item:iss}
  \item The trajectories of the averaged model (\ref{mainsystem2}) satisfy
    \begin{equation}
      \label{eq:avgmubnd}
      \|x(t)\| \le K \|x(0)\| e^{-\lambda t} + G \bar\mu
    \end{equation}
    with the same positive constants $K$, $\lambda$ and $G$ as in item~\ref{item:iss}) above, and provided $\mu(t) \in \setU$ for all $t\ge 0$.\label{item:avgub}
  \end{enumerate}
\end{theorem}
Theorem~\ref{thm:boundedness} establishes the global uniform boundedness and global uniform ultimate boundedness of the state of the averaged model (\ref{mainsystem2}) or (\ref{eq:avgmodel}), irrespectively of the duty cycle evolution but provided that $\mu(t) \in \setU$ for all $t\ge 0$. We highlight that the global nature of this result is only theoretical because it is based on the assumption that the inverter operates in CCM. As we previously mentioned, in order for the inverter to operate in CCM, then the instantaneous $C_1$ voltage cannot become lower than $-V_{\inw}$. 

The importance of Theorem~\ref{thm:boundedness} lies on the fact that it establishes the boundedness of the trajectories for \emph{every} possible evolution of the duty cycle (provided it does not reach the extreme values) and \emph{all} positive values of $L_1$, $L_2$, $C_1$, $C_2$ and $R$. Theorem~\ref{thm:boundedness} can also be used to show that the difference between any two trajectories starting from different initial conditions but corresponding to the same duty cycle evolution will ultimately converge to zero. 
\begin{cor}
  \label{cor:perconv}
  Let $0<\varepsilon\le 0.5$ and $\mu^* : \R \to \setU$ with $\setU$ as in (\ref{eq:setU}). Let $\phi(t,x_o)$ denote the solution to (\ref{mainsystem2}) that satisfies $\dot\phi(t,x_o) = A(\mu^*(t))\phi(t,x_o)+B\mu^*(t)$ and $\phi(0,x_o)=x_o$. Then, there exist positive constants $K$ and $\lambda$ such that for all $t\ge 0$ and all $x_o,y_o$, 
  \begin{equation}
    \label{eq:perconv}
    \| \phi(t,x_o) - \phi(t,y_o) \| \le Ke^{-\lambda t} \|x_o-y_o\|.
  \end{equation}
\end{cor}
\begin{proof}
  Let $\delta(t) = \phi(t,x_o)-\phi(t,y_o)$. Then, $\dot\delta(t) = A(\mu^*(t)) \delta(t)$ and apply Theorem~\ref{thm:boundedness}\ref{item:defneg}).
\end{proof}
Corollary~\ref{cor:perconv} becomes especially useful in the case of a periodic duty cycle, such as that proposed in \cite{5759746} and shown in (\ref{duty}). When $d(t)$ (and hence $\mu(t)$) is periodic and there exists a periodic steady-state trajectory $x^*(t)$ corresponding to such a periodic duty cycle evolution, then Corollary~\ref{cor:perconv} establishes that the state trajectory will converge to the steady-state periodic trajectory when started from \emph{any} initial condition (provided the inverter operates in CCM at all times). This latter result establishes the rationale for the inverter operation strategy proposed in \cite{5759746}, since it shows that the start-up transient will become extinct and hence the system will (eventually) operate along the periodic steady-state trajectory (which generates a sinusoidal waveform at the output if its frequency is low enough).

\section{Switched model stability}
\label{sec:sw-model-stab}

In this section, we provide results analogous to those in Theorem~\ref{thm:boundedness} and Corollary~\ref{cor:perconv} but for the switched model of the semi-quasi-Z-source inverter.

Consider again the switched model (\ref{eq:1})--(\ref{eq:P}). We define
\begin{equation}
  \label{eq:xswitch}
  x(t) := z(t) - \overline{\langle z \rangle}_{0.5},\quad
  \mu(t) := r(t) - 0.5
\end{equation}
with $\overline{\langle z \rangle}_{0.5}$ as in (\ref{eq:barz}). In the new variables, the switched model of the inverter coincides with (\ref{mainsystem2})--(\ref{eq:A0E}), the only difference being that $\mu(t)$ can only take the values $0.5$ (Mode I) or $-0.5$ (Mode II). We consider that the switching function $\mu(t)$ is the output of a PWM. We thus consider the following class of switching signals.
\begin{defn}
  A signal $\mu : \R \to \{-0.5,0.5\}$ is said to be of class $\PWM(T,\epsilon)$ with $0 < 2\epsilon \le T$, if it is right-continuous and for every integer $k$ there exists a time instant $t_k$ satisfying $kT+\epsilon \le t_k \le (k+1)T-\epsilon$ and
  \begin{equation}
    \label{eq:mupwm}
    \mu(t) =
    \begin{cases}
      0.5 &\text{if }\ kT \le t < t_k,\\
      -0.5 &\text{if }\ t_k \le t < (k+1)T.
    \end{cases}
  \end{equation}
  The quantity $T$ will be called the period of $\mu$ (even though $\mu$ need not be periodic) and $\epsilon$ the dwell-time.
\end{defn}

Our main results for the switched model are given below as Theorem~\ref{thm:Vin0GUES} and Corollary~\ref{cor:swtrstab}. These results require the following lemma.
%
%
\begin{lem}
  \label{lem:cqlf}
  Consider the matrix $P$ as in (\ref{eq:P}), and $A_{\I}$, $A_{\II}$ as in (\ref{eq:ArArq})--(\ref{eq:AIAII}). Let $t_{\I}$, $t_{\II}$ be positive and
  for every $\epsilon \ge 0$ define
  \begin{equation}
    \label{eq:M}
    M_\epsilon := e^{\D A_{\I}\epsilon}e^{\D A_{\II}t_{\II}}e^{\D A_{\I} t_{\I}}.
  \end{equation}
  \begin{enumerate}[a)]
  \item If $t_{\II} < \pi \sqrt{L_1C_1}$,\label{item:swdtcqlf}
    then $\srad(M_0)<1$ and 
    \begin{equation}
      \label{eq:Me}
      M_\epsilon' P M_\epsilon - P < 0,\quad\text{for all }\epsilon>0.
    \end{equation}
  \item If $t_{\II} = k\pi\sqrt{L_1C_1}$ for some positive integer $k$, then $\srad(M_\epsilon)=1$ for all $\epsilon\ge 0$.\label{item:swdtnexp}
  \end{enumerate}
\end{lem}

\begin{theorem}
  \label{thm:Vin0GUES}
  Consider $A(\cdot)$ and $B$ as in (\ref{eq:Amu})--(\ref{eq:A0E}) with $P$ as in (\ref{eq:P}) and let $\epsilon$ and $T$ satisfy $0 < 2\epsilon \le T < \pi\sqrt{L_1 C_1}$. Then, for all $\mu \in \PWM(T,\epsilon)$, we have that
  \begin{enumerate}[a)]
  \item there exist positive constants $K$ and $\lambda$ so that for every initial condition $x(0)$ and all $t\ge 0$, the trajectories of (\ref{eq:avgVin0}) satisfy
    \begin{equation}
      \label{eq:uges}
      \|x(t)\| \le K \|x(0)\| e^{-\lambda t},
    \end{equation}\label{item:swv0exp}
  \item there exist positive constants $K$, $\lambda$ and $G$ such that for every initial condition $x(0)$ and all $t\ge 0$, the trajectories of (\ref{eq:avgpert}) satisfy
    \begin{equation}
      \label{eq:ISSu}
      \|x(t)\| \le K \|x(0)\| e^{-\lambda t} + G \sup_{0 \le \tau \le t} \|u(\tau)\|,
    \end{equation}\label{item:swiss}
  \item the trajectories of the switched model (\ref{mainsystem2}) satisfy
    \begin{equation}
      \label{eq:uub}
      \|x(t)\| \le K \|x(0)\| e^{-\lambda t} + \frac{G}{2}
    \end{equation}
    with the same $K$, $\lambda$ and $G$ as for item~\ref{item:swiss}) above.\label{item:swbnd}
  \end{enumerate}
\end{theorem}
\begin{cor}
  \label{cor:swtrstab}
  Consider $A(\cdot)$ and $B$ as in (\ref{eq:Amu})--(\ref{eq:A0E}) with $P$ as in (\ref{eq:P}). Let $\epsilon$ and $T$ satisfy $0 < 2\epsilon \le T < \pi\sqrt{L_1 C_1}$ and consider $\mu^* \in \PWM(T,\epsilon)$. Let $\phi(t,x_o)$ denote the solution to (\ref{mainsystem2}) that satisfies $\dot\phi(t,x_o) = A(\mu^*(t))\phi(t,x_o)+B\mu^*(t)$ and $\phi(0,x_o)=x_o$. Then, there exist positive constants $K$ and $\lambda$ such that for all $t\ge 0$ and all $x_o,y_o$, 
  \begin{equation}
    \label{eq:swperconv}
    \| \phi(t,x_o) - \phi(t,y_o) \| \le Ke^{-\lambda t} \|x_o-y_o\|.
  \end{equation}
\end{cor}
\begin{proof}
  Let $\delta(t) = \phi(t,x_o)-\phi(t,y_o)$. Then, $\dot\delta(t) = A(\mu^*(t)) \delta(t)$ and apply Theorem~\ref{thm:Vin0GUES}\ref{item:swv0exp}).
\end{proof}
Theorem~\ref{thm:Vin0GUES}\ref{item:swv0exp}) establishes the global exponential stability of the switched system (\ref{eq:avgVin0}), uniformly over switching signals of PWM type, provided the PWM period $T$ is suitably limited. Since neither $A(-0.5)=A_{\II}$ nor $A(0.5)=A_{\I}$ are Hurwitz, then the exponential stability of the switched system (\ref{eq:avgVin0}) is not possible under arbitrary switching, nor by imposing only a minimum dwell-time on each mode \cite{linant_tac09,showir_siamrev07}. By Lemma~\ref{lem:cqlf}\ref{item:swdtnexp}), it follows that imposing a \emph{maximum} dwell-time for Mode II is essential in order to ensure the exponential stability.

The switched system (\ref{eq:avgVin0}) has a natural quadratic Lyapunov function given by the circuit's stored energy: $V(x) = x'Px$ with $P$ as in (\ref{eq:P}). The derivative of this energy function along the trajectories of (\ref{eq:avgVin0}) has the form $-x'C'Cx$, where
\begin{align*}
  A_{\I}'P+PA_{\I} &= A_{\II}'P+PA_{\II} = -C'C,\\
  C &=
  \begin{bmatrix}
    0 & 0 & 0 & R^{-1/2}
  \end{bmatrix},
\end{align*}
and is hence only negative semidefinite. Even though several existing results extend LaSalle's Invariance Principle to switched systems \cite{hespan_tac04,bacmaz_scl05,chewan_tac08}, none of these can be employed to ascertain the uniform asymptotic or exponential stability of (\ref{eq:avgVin0}) by means of such natural Lyapunov function. The main difficulty is that neither $(A_{\I},C)$ nor $(A_{\II},C)$ are detectable, as can be easily checked. 
Theorem~8 of \cite{hespan_tac04} relaxes this observability assumption. When applied to system (\ref{eq:avgVin0}) employing the natural quadratic Lyapunov function $V(x)$, Theorem~8 of \cite{hespan_tac04} establishes that for a sufficiently large type of switching signals (identified as $\Sb_{p\text{-dwell}}$ in \cite{hespan_tac04}), the trajectories of (\ref{eq:avgVin0}) will converge to the smallest subspace $\mathcal{M}$ that is invariant under both $A_{\I}$ and $A_{\II}$, and contains both the unobservable subspaces of $(A_{\I},C)$ and $(A_{\II},C)$. The problem in the current case is that $\mathcal{M}=\R^4$, which coincides with the whole state space. In view of the aforementioned difficulties, one main feature of our proof of Theorem~\ref{thm:Vin0GUES} is that all stability results are established by exploiting the natural Lyapunov function $V(x)=x'Px$.

Analogously to Theorem~\ref{thm:boundedness} in Section~\ref{sec:stab-results}, Theorem~\ref{thm:Vin0GUES}\ref{item:swiss}) establishes that the perturbed switched system (\ref{eq:avgpert}) is ISS with respect to the perturbation input $u$, having linear input-to-state gain and exponential decay, uniformly over switching signals in $\PWM(T,\epsilon)$.
Theorem~\ref{thm:Vin0GUES}\ref{item:swbnd}) establishes that all trajectories of the inverter circuit remain bounded, and ultimately bounded uniformly over the switching signals considered. 
Finally, Corollary~\ref{cor:swtrstab} shows that the difference between trajectories corresponding to different initial conditions but the same switching signal will exponentially decrease to zero, and that the exponential convergence rate can be uniform over all the switching signals considered. 
As previously mentioned in the case of the averaged model, this result gives theoretical justification to the inverter operation strategy proposed in \cite{5759746}.

\section{Conclusions}
\label{sec:conclusions}

We have provided large-signal stability results for both the standard averaged model and the switched model of the semi-quasi-Z-source inverter operating in continuous conduction mode. 
Specifically, we have established that the state trajectories are bounded and ultimately bounded with exponential decay rate, uniformly over all duty cycle evolutions provided reasonable operating conditions are ensured. 
Our stability results hold for every possible value of the circuit inductors, capacitors and linear resistive load, and can be used to show that all trajectories corresponding to the same duty cycle evolution but different initial conditions will converge to the same steady-state trajectory. 
This result gives theoretical justification to the inverter operation strategy proposed in \cite{5759746}, in the sense that under reasonable restrictions on the PWM signal, the state trajectory will converge exponentially to the steady state one from any initial condition. 
Our stability results are derived by exploiting the natural energy function of the circuit, although the quadratic term in its derivative along the system trajectories is only negative semidefinite.

Future work will be focused on establishing similar stability results for more general, including nonlinear, load types, and on obtaining less conservative transient and ultimate bounds.

\appendix

\section*{Proof of Theorem~\ref{thm:boundedness}}

We require the following lemma, whose proof is straightforward and hence omitted.
%
\begin{lem}
  \label{lem:HkerQ}
  Let $\bar A$ be Hurwitz and suppose that there exists $P=P' > 0$ such that:
  \begin{equation}
    \label{condicionP}
    \bar A'P + P \bar A = -Q \leq 0.
  \end{equation}
  Then, there exists $\tilde{P}=\tilde{P}' > 0$ such that $\bar A'\tilde{P} + \tilde{P} \bar A < 0$
  if and only if there exists $H=H'$ satisfying:
  \begin{align}
    \label{eq:XSQ}
    X(\bar A) &:= S_{Q}'(\bar A'H+H \bar A)S_{Q} < 0,
  \end{align}
where the columns of $S_Q$ constitute an orthonormal basis for the null space of $Q$, i.e. $S_Q \in \R^{n\times r}$, $S_Q'S_Q = \id$, $QS_Q=0$, and $r$ is the greatest such integer. If \eqref{eq:XSQ} holds, then every $\tilde{P}$ of the form $\tilde{P} = P + \xi H$ with $\xi > 0$ sufficiently small will satisfy $\bar A'\tilde{P} + \tilde{P} \bar A < 0$.
\end{lem}

\ref{item:defneg}) First, note that $A(\mu)$ is Hurwitz for all $\mu\in\setU$. The proof is based on showing that $X(A(\mu)) < 0$ for all $\mu\in\setU$ for the given matrix $H$. Note that for $\bar A = A(\mu)$, (\ref{condicionP}) is satisfied with $P$ as in (\ref{eq:P}) for all $\mu\in\setU$. We have $\ker Q = \img S_Q$, $S_Q = [e_1, e_2, e_3]$. Let $X_i$ denote the principal $i$-th order determinant of $X$, so that $X<0$ if and only if $X_i < 0$ for odd $i$ and $X_i > 0$ for even $i$. Direct computation gives
\begin{equation*}
  X(A(\mu))_1 = \frac{h_{13} (2\mu + 1)}{C_1} < 0, 
\end{equation*}
for all $\mu\in\setU$, where the inequality follows from (\ref{eq:Hentries1}). Also,
\begin{align*}
  X(A(\mu))_2 &= - \frac{b}{C_1} - \frac{2 X(A(\mu))_1 h_{24}}{C_2},\\
  b &= \frac{(h_{13} - h_{23})^2 \mu^2 + h_{13}h_{23}}{C_1},
\end{align*}
and $X(A(\mu))_2 > 0$ for all $\mu\in\setU$ follows from (\ref{eq:Hentries3}). Finally,
\begin{align*}
  X(A(\mu))_3 &=\det(X(A(\mu))) = -c\,X(A(\mu))_2 \\
  c &= \frac{L_2 h_{13} (2\mu+1) + L_1 h_{23} (2\mu-1)}{L_1 L_2}.
\end{align*}
We have $c>0$ for all $\mu\in\setU$ from (\ref{eq:Hentries2}), and hence $X(A(\mu))_3 > 0$ for all $\mu\in\setU$.

By application of Lemma~\ref{lem:HkerQ}, we know that there exist $\xi_1,\xi_2 > 0$ such that $A(\bar\mu)'\tilde P_{\xi} + \tilde P_{\xi} A(\bar\mu) < 0$ for all $0 < \xi < \xi_1$ and $A(-\bar\mu)'\tilde P_{\xi} + \tilde P_{\xi} A(-\bar\mu) < 0$ for all $0 < \xi < \xi_2$. Since $A(\mu)$ is a convex combination of $A(-\bar\mu)$ and $A(\bar\mu)$, it follows that $A(\mu)'\tilde P_{\xi} + \tilde P_{\xi} A(\mu) < 0$ for all $\mu\in\setU$ and all $0 < \xi < \min\{\xi_1,\xi_2\}$.

\ref{item:iss}) Let $\xi>0$ be such that $\tilde P_\xi > 0$ and (\ref{eq:VdotVin0}) holds for all $\mu\in\setU$. Along trajectories of (\ref{eq:avgpert}) we have that for all $\mu\in\setU$,
\begin{equation*}
  \dot V_\xi(x,\mu,u) \le -\alpha(\xi) \|x\|^2 + 2 \|\tilde P_\xi B\| \|x\| \|u\|,
\end{equation*}
and the proof follows from standard ISS arguments \cite{khalil02}.

\ref{item:avgub}) Just apply Theorem~\ref{thm:boundedness}\ref{item:iss}) with $u(t)=\mu(t)$. \hfill\QED

\section*{Proof of Lemma~\ref{lem:cqlf}}

  We have
  \begin{align}
    A_{\I}' P + P A_{\I} = A_{\II}' P + P A_{\II} = -Q \le 0,
  \end{align}
  with
  \begin{equation}
    \label{eq:Q}
    Q = \diag\left(0,0,0,\frac{1}{R}\right).
  \end{equation}
  Consequently, for every $\epsilon \ge 0$ we can prove that
  \begin{equation}
    \label{eq:MdeltaV}
    M_\epsilon'PM_\epsilon - P \le 0,\quad\quad (M_\epsilon^2)' P M_\epsilon^2 - P \le 0,
  \end{equation}
  and
  \begin{multline}
    \label{eq:le0p}
    x'(M_\epsilon'PM_\epsilon-P)x < 0\quad\Longrightarrow\\
    x'((M_\epsilon^2)'PM_\epsilon^2 - P)x < 0
  \end{multline}

  \ref{item:swdtcqlf}) Let $x_o$ satisfy
  \begin{equation}
    \label{eq:xkerQ}
    x_o'(M_\epsilon'PM_\epsilon-P)x_o = 0.
  \end{equation}
  The vector $x_o$ satisfies (\ref{eq:xkerQ}) if and only if
  \begin{align}
    \label{eq:invI}
    e^{A_{\I} t}x_o &\in \Sb_{\I},\quad\text{for all }t\in[0,t_{\I}),\\
    \label{eq:invII}
    e^{A_{\II}t} e^{A_{\I} t_{\I}}x_o &\in \Sb_{\II},\quad \text{for all }t\in[0,t_{\II})\quad\text{and}\\
    \label{eq:forcontr}
    e^{A_{\I} t} e^{A_{\II}t_{\II}} e^{A_{\I} t_{\I}}x_o &\in \Sb_{\I},\quad \text{for all }t\in[0,\epsilon).
  \end{align}
  where $\Sb_{\I} = \spn\{e_1\}$ and $\Sb_{\II}=\spn\{e_1,e_3\}$ are the largest subspaces invariant under $A_{\I}$ or $A_{\II}$, respectively, that are contained in $\ker Q$. From (\ref{eq:ArArq})--(\ref{eq:AIAII}) and (\ref{eq:invI}), it follows that $x_o = e^{A_{\I} t}x_o = \alpha e_1$ for some $\alpha\in\R$ and for all $t\in[0,t_{\I}]$. From (\ref{eq:ArArq})--(\ref{eq:AIAII}) and (\ref{eq:invII}), then $e^{A_{\II}t_{\II}}e^{A_{\I} t_{\I}}x_o = \alpha [\cos(\omega t_{\II}) e_1 + \sqrt{L_1/C_1} \sin(\omega t_{\II})e_3]$, with $\omega = 1/\sqrt{L_1C_1}$. Since $t_{\II} < \pi \sqrt{L_1C_1}$, then $\sin(\omega t_{\II}) \neq 0$. Hence, if $\epsilon>0$, from (\ref{eq:forcontr}) we must have $x_1 := M_0 x_o = e^{A_{\II}t_{\II}}e^{A_{\I} t_{\I}}x_o \in \Sb_{\I}$, which implies that $x_o=0$. Therefore, if $\epsilon>0$, (\ref{eq:xkerQ}) implies that $x_o=0$ and (\ref{eq:Me}) is established.

  For $\epsilon=0$, suppose that $x_o \neq 0$ satisfies (\ref{eq:xkerQ}). By the previous argument, $x_1 = M_0 x_o \notin \Sb_{\I}$ and 
  \begin{align*}
    x_1'(M_0'PM_0-P)x_1 &= x_o' ((M_0^2)' P M_0^2 - M_0PM_0) x_o\\
    &= x_o' ((M_0^2)'PM_0^2 - P) x_o < 0,
  \end{align*}
  where the latter inequality follows because, repeating the previous
  argument, $x_1'(M_0'PM_0 - P)x_1 = 0$ would imply that $x_1 \in
  \Sb_1$, a contradiction. Recalling~(\ref{eq:le0p}), then we conclude
  that $\srad(M_0^2)<1$ and hence $\srad(M_0)<1$.

  \ref{item:swdtnexp}) By (\ref{eq:MdeltaV}), then $\srad(M_\epsilon) \le 1$ for all $\epsilon\ge 0$. Take $x_o = e_1$. Then, $e^{A_{\I}t_{\I}}x_o = x_o$ and $e^{A_{\II}t_{\II}}x_o = \cos(\omega t_{\II}) e_1 + \sqrt{L_1/C_1} \sin(\omega t_{\II})e_3 = (-1)^k e_1 = (-1)^k x_o$. Therefore, $M_\epsilon x_o = (-1)^k x_o$, showing that $\lambda=(-1)^k$ is an eigenvalue of $M_\epsilon$ with $|\lambda|=1$. \hfill\QED

\section*{Proof of Theorem~\ref{thm:Vin0GUES}}

Consider $V(x) = x'Px$. For all $x$, we have
  \begin{align}
    \label{eq:Vineq}
    \kappa_1 \|x\|^2 &\le V(x) \le \kappa_2 \|x\|^2,\\
    \label{eq:Vdot}
    \dot V(x,\mu) &= \dot V(x) = -x' Q x \le 0
  \end{align}
  with $\kappa_1,\kappa_2 > 0$ and $Q$ as in (\ref{eq:Q}) being positive semidefinite. Consider time instants of the form $t_k = kT + \epsilon/2$ for all nonnegative integers $k$. 

\ref{item:swv0exp}) We have
  \begin{align}
    \label{eq:xtItIItI}
    x(t_{k+1}) &= M(t_{\I,k},t_{\II,k}) x(t_k),\\
    \label{eq:Mk}
    M(t_{\I,k},t_{\II,k}) &:= e^{\displaystyle A_{\I}\epsilon/2} e^{\displaystyle A_{\II}t_{\II,k}} e^{\displaystyle A_{\I} t_{\I,k}},
  \end{align}
  with $t_{\II,k} \in [\epsilon, T-\epsilon]$ and $t_{I,k} = T - t_{\II,k} - \epsilon/2$  because $\mu \in \PWM(T,\epsilon)$. Since $T < \pi\sqrt{L_1C_1}$, using Lemma~\ref{lem:cqlf} it follows that
  \begin{equation}
    \label{eq:MPM}
    R(t_{\I,k},t_{\II,k}) := M(t_{\I,k},t_{\II,k})' P M(t_{\I,k},t_{\II,k}) - P < 0.
  \end{equation}
  Define
  \begin{equation}
    \label{eq:lambda}
    \kappa_3 := \sup_{
      \begin{smallmatrix}
        t_{\II} \in [\epsilon, T-\epsilon]\\
        t_{\I} = T- t_{\II} - \epsilon/2
      \end{smallmatrix}}
    \lambda_{\min}[-R(t_{\I},t_{\II})]
  \end{equation}
  By (\ref{eq:MPM}), and since $M(\cdot,\cdot)$ is continuous on its arguments and the supremum above is taken over a compact set, then $\kappa_3 > 0$. We thus have
  \begin{align}
    \lefteqn{V(x(t_{k+1})) - V(x(t_k)) =}\notag\\
    & x(t_k)' \left[M(t_{\I,k},t_{\II,k})' P M(t_{\I,k},t_{\II,k}) - P\right] x(t_k) \notag\\
    \label{eq:incVV}
    & \le -\kappa_3 \| x(t_k) \|^2 \le -\frac{\kappa_3}{\kappa_2} V(x(t_k)),\quad\text{for all }k\ge 0.
  \end{align}
  Consequently, $V(x(t_k)) \le (1-\frac{\kappa_3}{\kappa_2})^k V(x(t_0))$. Taking (\ref{eq:Vdot}) and (\ref{eq:incVV}) into account, we can obtain the following bound
  \begin{align}
    \label{eq:boundVt}
    V(x(t)) &\le e^{r(\epsilon/2+T)} e^{-rt} V(x(0)),\quad\text{for all }t\ge 0,\\
    r &:= -\frac{\log\left(1 - \kappa_3/\kappa_2 \right)}{T} > 0,
  \end{align}
  whence, for all $t\ge 0$,
  \begin{equation}
    \label{eq:normbound}
    \|x(t)\| \le \sqrt{\frac{\kappa_2}{\kappa_1} e^{r(\epsilon/2+T)}} e^{-rt/2} \|x(0)\|.
  \end{equation}
  Therefore, we have established~(\ref{eq:uges}) with $K=\sqrt{\frac{\kappa_2}{\kappa_1} e^{r(\epsilon/2+T)}}$ and $\lambda=r/2$.

  \ref{item:swiss})
  Let $\Phi(t;t_{\I},t_{\II},\epsilon)$ be defined as follows
  \begin{equation}
    \label{eq:Phi}
    {\scriptsize
    \Phi(t;t_{\I},t_{\II},\epsilon) :=
    \begin{cases}
      e^{\D A_{\I} t} &\text{if }0\le t < t_{\I},\\
      e^{\D A_{\II}(t-t_{\I})} e^{\D A_{\I} t_{\I}} &\text{if }t_{\I} \le t < t_{\I} + t_{\II},\\
      e^{\D A_{\I} t} e^{\D A_{\II} t_{\II}} e^{\D A_{\I} t_{\I}} &\text{if }t_{\I}+t_{\II} \le t \le t_{\I} + t_{\II} + \epsilon.
    \end{cases}}
  \end{equation}
  For all $t_k \le t \le t_{k+1}$, we have
  \begin{align}
    \label{eq:xtk}
    x(t) &= \Phi(t-t_k;t_{\I,k},t_{\II,k},\epsilon/2) x(t_k) + \Psi_k(t),\ \text{with} \\
    \label{eq:Psik}
    \Psi_k(t) &:= \int_{t_k}^{t} \Phi(t-\tau;t_{\I,k},t_{\II,k},\epsilon/2) B u(\tau) d\tau,
  \end{align}
  where $t_{\II,k} \in [\epsilon,T-\epsilon]$ and $t_{\I,k} = T- t_{\II,k} -\epsilon/2$ for all $k$. For all $t_k \le t \le t_{k+1}$, we can bound the norm of $\Psi_k(t)$ as follows
  \begin{align}
    \label{eq:normPsik}
    \|\Psi_k(t)\| 
    &\le \bar\Phi \|B\| \bar u_k,\\
    \label{eq:barphi}
    \bar\Phi &:= \sup_{
      \begin{smallmatrix}
        t \in [0,T]\\
        t_{\II} \in [\epsilon, T-\epsilon]\\
        t_{\I} = T- t_{\II} - \epsilon/2
      \end{smallmatrix}} 
    \|\Phi(t;t_{\I},t_{\II},\epsilon/2)\|,\\
    \label{eq:nintuk}
    \bar u_k &:= \int_{t_k}^{t_{k+1}} \|u(t)\| dt.
  \end{align}
  Using (\ref{eq:xtk})--(\ref{eq:nintuk}), we can bound the increment of $V$ as
  \begin{multline}
    \label{eq:deltaVwu}
    V(x(t_{k+1})) - V(x(t_k))\\
    \le -\kappa_3 \|x(t_k)\|^2 + c_1 \bar u_k \|x(t_k)\| + c_2 \bar u_k^2
  \end{multline}
  with 
  $\kappa_3$ as in (\ref{eq:Mk})--(\ref{eq:lambda}) and $c_1,c_2>0$. It follows that for every $0<\theta<1$,
  \begin{align}
    V(x(t_{k+1})) - V(x(t_k)) &\le -\kappa_3 (1-\theta) \|x(t_k)\|^2\notag\\ 
    \label{eq:Vpertni}
    &\le -\frac{\kappa_3(1-\theta)}{\kappa_2} V(x(t_k)).
  \end{align}
  whenever
  \begin{equation*}
    \|x(t_k)\| \ge \frac{c_1 + \sqrt{c_1^2 + 4\kappa_3\theta c_2}}{2\kappa_3\theta} \bar u_k =: L_1 \bar u_k. 
  \end{equation*}
  Following standard Lyapunov arguments \cite{khalil02}, we can reach the following inequality
  \begin{equation}
    \label{eq:normxtkb}
    \|x(t_k)\| \le c_3 \alpha^k \|x(t_0)\| + L_2 \sup_{0\le \tilde k \le k-1} \bar u_{\tilde k},
  \end{equation}
  where $0< \alpha < 1$.
  From (\ref{eq:xtk}), it follows that
  \begin{equation}
    \label{eq:nxbtk}
    \|x(t)\| \le \bar\Phi \|x(t_k)\| + \bar\Phi \|B\| \int_{t_k}^t \|u(\tau)\| d\tau 
  \end{equation}
  for all $t_k \le t \le t_{k+1}$ for all $k \ge 0$.
  Combining (\ref{eq:normxtkb})--(\ref{eq:nxbtk}),
  \begin{multline}
    \label{eq:nxtxt0}
    \|x(t)\| \le \bar\Phi c_3 \alpha^k \|x(t_0)\|\\
    + \bar\Phi L_2 \sup_{0\le \tilde k \le k-1} \bar u_{\tilde k} + \bar\Phi \|B\| \int_{t_k}^t \|u(\tau)\| d\tau,
  \end{multline}
  for all $t_k \le t \le t_{k+1}$. Define
  \begin{equation}
    \label{eq:normu}
    \|u_{[a,b]}\| := \sup_{\tau \in [a,b]} \|u(\tau)\|.
  \end{equation}
  Then,
  \begin{equation}
    \label{eq:relnorm}
    \bar u_k \le T \|u_{[t_k,t_{k+1}]}\| \le T \|u_{[t_0,t_{k+1}]}\|
  \end{equation}
  Using (\ref{eq:relnorm}) in (\ref{eq:nxtxt0}) we arrive at
  \begin{equation}
    \label{eq:nxtnu}
    \|x(t)\| \le \bar\Phi c_3 \alpha^k \|x(t_0)\| + c_5 \|u_{[t_0,t]}\|
  \end{equation}
  for all $t_k \le t \le t_{k+1}$, for all $k\ge 0$.
  Therefore,
  \begin{equation}
    \label{eq:nxtt0exp}
    \|x(t)\| \le c_6 e^{-\lambda t} \|x(t_0)\| + c_5 \|u_{[t_0,t]}\|
  \end{equation}
  with $\lambda > 0$. Since $t_0 = \epsilon/2$, it follows that
  \begin{equation}
    \label{eq:nxtnu2}
    \|x(t_0)\| \le \|e^{A_{\I} \epsilon/2}\| \|x(0)\| + \bar\Phi \|B\| \|u_{[0,\epsilon/2]}\| \epsilon/2 .
  \end{equation}
  Combining (\ref{eq:nxtt0exp})--(\ref{eq:nxtnu2}), the result follows.

  \ref{item:swbnd}) Just apply Theorem~\ref{thm:Vin0GUES}\ref{item:swiss}) with $u(t) = \mu(t)$.\hfill\QED

\bibliographystyle{IEEEtran}

\small{
\bibliography{DNL}
}
\end{document}